\documentclass{llncs}
\usepackage{graphicx,epstopdf}
\usepackage{amsmath}

\DeclareMathOperator{\KS}{C}

\DeclareMathOperator{\poly}{poly}
\DeclareMathOperator{\polylog}{polylog}

\DeclareMathOperator{\prob}{Prob}

\DeclareMathOperator{\Ext}{Ext}

\DeclareMathOperator{\Con}{Con}

\newcommand{\U}{\mathcal{U}}
\newcommand{\V}{\mathcal{V}}
\newcommand{\eps}{\varepsilon}

\begin{document}

\title{Improving the Space-Bounded Version of Muchnik's Conditional Complexity Theorem via ``Naive'' Derandomization\thanks{Supported by ANR Sycomore, NAFIT ANR-08-EMER-008-01 and RFBR~09-01-00709-a grants.}
}
\titlerunning{Naive derandomization}
\author{Daniil Musatov}
\institute{Lomonosov Moscow State University,
            \email{musatych@gmail.com}
}

\date{}
\maketitle

\begin{abstract}
Many theorems about Kolmogorov complexity rely on existence of combinatorial objects with specific properties. Usually the probabilistic method gives such objects with better parameters than explicit constructions do. But the probabilistic method does not give ``effective'' variants of such theorems, i.e. variants for resource-bounded Kolmogorov complexity. We show that a ``naive derandomization'' approach of replacing these objects by the output of Nisan-Wigderson pseudo-random generator may give polynomial-space variants of such theorems.

Specifically, we improve the preceding polynomial-space analogue of Muchnik's conditional complexity theorem. I.e., for all $a$ and $b$ there exists a program $p$ of least possible length that transforms $a$ to $b$ and is simple conditional on $b$. Here all programs work in polynomial space and all complexities are measured with logarithmic accuracy instead of polylogarithmic one in the previous work.
 \end{abstract}

\section{Introduction}
Many statements about Kolmogorov complexity may be proven by applying some combinatorial constructions like expanders or extractors. Usually these objects are characterized by some parameters, and one may say which parameters are ``better''. Very often the probabilistic method allows one to obtain these objects with much better parameters than explicit constructions do. But exploiting the probabilistic method causes exponential-space brute-force search for an object satisfying the necessary property. And if this search is performed while describing some string to obtain an upper bound on its complexity then this bound cannot be repeated for a resource-bounded version of complexity, even for the polynomial-space one. On the other hand, replacing the probabilistic method by an explicit construction weakens the statement due to worse parameters. 

We present a technique that combines advantages of both probabilistic and explicit construction methods. The key idea is to substitute a random object with a pseudo-random one still possessing the necessary property. The employed property of a pseudo-random generator is the indistinguishability one: its output cannot be distinguished from a random string by boolean circuits of constant depth and polynomial size. The Nisan-Wigderson generator satisfies this condition. If the necessary property can be tested by such circuits then it holds for a pseudo-random object with approximately the same probability as for a truly random object, so a search among all seeds may be perormed. Unfortunately, it is not clear how to build such a circuit for the extractor property and similar ones, so we relax the property in a way that allows both proving the theorem and building polynomial constant-depth circuits.

This ``naive derandomization'' idea has been recently applied by Andrei Romashchenko (\cite{romash}) in another situation: a probabilistic bit-probe scheme with one-sided error for the membership problem is constructed.

By exploiting the new method we improve the previous result~\cite{tocs} generalizing Muchnik's conditional complexity theorem. The original theorem~\cite{muchnik-codes} states that for all $a$ and $b$ of length $n$ there exists a program $p$ of length $\KS(a|b)+O(\log n)$ that transforms $b$ to $a$ and has complexity $O(\log n)$ conditional on $a$. In~\cite{tocs} this result is restated for space-bounded complexity with gap rised from $O(\log n)$ to $O(\log^3 n)$. The main idea was to employ a property of extractors, proven in~\cite{fortnow}: in an extractor graph in any sufficiently big subset $S$ of the left part there are few vertices with all right-part neighbours having indegree from $S$ twice greater than average. We refer to this property as to ``low-congesting'' one. Best explicit extractor constructions yield a space-bounded version of the theorem with polylogarithmic precision. In this paper we replace an explicit extractor by a pseudo-random graph, that does not necessary have the extractor property, but enjoys the low-congesting property for ``relevant'' subsets $S$. This replacement leads to decreasing the precision back to logarithmic term.

The rest of the paper is organized as follows. In Sect.~\ref{pre} we give formal definitions of all involved objects and formulate necessary results. In Sect.~\ref{muchnik} we formally state our space-bounded variant of Muchnik's theorem and specify all details of the proof. 

\section{Preliminaries}\label{pre}
\subsection{Kolmogorov complexity}
Let $\V$ be a two-argument Turing machine. We refer to the first argument as to the ``program'' and to the second argument as to the ``argument''. (Plain) Kolmogorov complexity of a string $x$ conditional on $y$ with respect to $\V$ is the length of a minimal program $p$ that transforms $y$ to $x$, i.e. 
$$\KS_{\V}(x\mid y)=\min\{p\colon \V(p,y)=x\}$$
There exists an optimal machine $\U$ that gives the least complexity up to an additive term. Specifically, $\forall\V\,\exists c\,\forall x,y \KS_{\U}(x|y)<\KS_{\V}(x|y)+c$. We employ such a machine $\U$, drop the subscript and formulate all theorems up to a constant additive term. The unconditional complexity $\KS(x)$ is the complexity with empty condition $\KS(x\mid\varepsilon)$, or the length of a shortest program \textit{producing} $x$.

Now we define the notion of resource-bounded Kolmogorov complexity. Roughly speaking, it is the length of a minimal program that transforms $y$ to $x$ efficiently. Formally, Kolmogorov complexity of a string $x$ conditional on $y$ in time $t$ and space $s$ with respect to $\V$ is the length of a shortest program $p$ such that $\V(p,y)=x$, and the computation of $\V(p,y)$ works in $t$ steps and uses $s$ cells of memory. This complexity is denoted by $\KS^{t,s}_{\V}(x\mid y)$. Here the choice of $\V$ alters not only complexity, but also time and space bounds. However, there still exists an optimal machine in the following sense:
\begin{theorem}
There exists a machine $\U$ such that for any machine $\V$ there exists a constant $c$ such that for all $x$, $y$, $s$ and $t$ it is true that $\KS^{s,t}_{\U}(x\mid y)\leq\KS^{cs,ct\log t}_{\V}(x\mid y)+c$.
\end{theorem}
This paper deals with space bounds only, so the time-bound superscript in all notations is dropped. Also, the machine subscript is dropped as before and all theorems are formulated with a constant additive term in all complexities and a constant multiplicative term in all space bounds.

\subsection{Extractors}
An extractor is a function that extracts randomness from weak random sources. A $k$-weak random source of length $n$ is a probabilistic distribution on $\{0,1\}^n$ with minentropy greater than $k$. The last condition means that no particular string occurs with probability greater than $2^{-k}$. An extractor with parameters $n$, $m$, $d$, $k$, $\eps$ is a function $\Ext\colon\{0,1\}^n\times\{0,1\}^d\to\{0,1\}^m$, such that for any independent $k$-weak random source $x$ of length $n$ and uniform distribution $u$ on $\{0,1\}^d$ the induced distribution $\Ext(x,u)$ on $\{0,1\}^m$ is $\eps$-close to uniform, that is, for any set $Y\subset\{0,1\}^m$ its probability differs from its fraction by at most $\eps$. This is interpreted as follows: an extractor receives $n$ weakly random bits and $d$ truly random bits independent from the first argument and outputs $m$ almost random bits.

Like any two-argument function, an extractor may be viewed as a bipartite (multi-)graph: the first argument indexes a vertex in the left part, the second argument indexes an edge going from this vertex, and the value indexes a vertex in the right part which this edge directs to. That is, the graph has $N=2^n$ vertices on the left, $M=2^m$ vertices on the right, and all left-part vertices have degree $D=2^d$. Throughout the paper, we say that a bipartite graph has parameters ($n$, $m$, $d$) if the same holds for it. For the sake of clarity we usually omit these parameters in extractor specifications. The extractor property may be also formulated in terms of graphs: for any left-part subset $S$ of size greater than $K=2^k$ and for any right-part subset $Y$ the fraction of edges directing from $S$ to $Y$ among all edges directing from $S$ differs from the fraction of vertices from $Y$ among all right-part vertices by at most $\eps$. Put formally, $\left||Y|/M-E(S,Y)/E(S,M)\right|<\eps$, where $E(S,T)$ is the number of edges diredting from $S$ to $T$ and $M$ slightly abusively denotes both the right part and the number of vertices in it. A proof of equivalence may be found in~\cite{extractor-bounds}.

It is proven by the probabilistic method (see, for example, \cite{extractor-bounds}) that for all $n$, $k$ and $\eps$ there exists an extractor with parameters $d=\log(n-k)+2\log(1/\eps)+O(1)$ and $m=k+d-2\log(1/\eps)-O(1)$. Nevertheless, no explicit (that is, running in polynomial time) construction of such an extractor is known. Best current results (\cite{extractor-explicit}, \cite{trevisan}) for $m=k$ use $d=\polylog n$ truly random bits. A number of explicit extractors including those in \cite{wigderson}, \cite{dvir} and \cite{guruswami} use $O(\log n)$ truly random bits but output only $(1-\alpha)k$ almost random bits. Such constructions are insufficient for our goals.

\subsection{Nisan-Wigderson generators}
The Nisan-Wigderson pseudo-random generator is a deterministic polynomal-time function that generates $n$ pseudorandom bits from $\polylog(n)$ truly random bits. The output of such a generator cannot be distinguished from a truly random string by small circuits. Specifically, we exploit the following theorem from~\cite{NW97}. (The statement was initially proven by Nisan in paper~\cite{Nisan91}).
\begin{theorem}\label{nw}
For any constant $d$ there exists a family of functions $G_n\colon \{0,1\}^k\to\{0,1\}^n$, where $k=O(\log^{2d+6} n)$, such that two properties hold:
\begin{description}
\item [Computability:] $G$ is computable in workspace $\poly(k)$;
\item [Indistinguishability:] For any family of circuits $C_n$ of size $\poly(n)$ and depth $d$ for any positive polynomial $p$ for all large enough $n$ it holds that:
$$\left|\prob_x\{C_n(G_n(x))=1\}-\prob_y\{C_n(y)=1\}\right|<\frac1{p(n)},$$
where $x$ is distributed uniformly in $\{0,1\}^k$ and $y$~--- in $\{0,1\}^n$.
\end{description}
\end{theorem}

By rescaling the parameters we get the following
\begin{corollary}\label{nw-modified}
For any constant $d$ there exists a family of functions $G_n\colon \{0,1\}^k\to\{0,1\}^N$, where $k=\poly(n)$ and $N=2^{O(n)}$, such that two properties hold:
\begin{itemize}
\item $G$ is computable in workspace $\poly(n)$;
\item For any family of circuits $C_n$ of size $2^{O(n)}$ and depth $d$, for any constant $c$ and for all large enough $n$ it holds that:
$$\left|\prob_x\{C_n(G_n(x))=1\}-\prob_y\{C_n(y)=1\}\right|<2^{-cn}.$$
\end{itemize}
Here all constants in $\poly$- and $O$-notations may depend on $d$ but do not depend on $k$, $n$ and $C_n$.
\end{corollary}

The last corollary implies the following basic principle:
\begin{lemma}\label{mainprinciple}
Let $\mathcal{C}_n$ be some set of combinatorial objects encoded by boolean strings of length $2^{O(n)}$. Let $\mathcal{P}$ be some property satisfied for fraction at least $\alpha$ fraction of objects in $\mathcal{C}_n$ that can be tested by a family of circuits of size $2^{O(n)}$ and constant depth. Then for sufficiently large $n$ the property $\mathcal{P}$ is satisfied for a fraction at least $\alpha/2$ of values of $G_n$, where $G_n$ is the function from the previous corollary.
\end{lemma}

\subsection{Constant-depth circuits for approximate counting}
It is well-known that constant-depth circuits cannot compute the majority function. Moreover, they cannot compute a general threshold function that equals $1$ if and only if the fraction of $1$'s in its input exceeds some threshold $\alpha$. Nevertheless, one can build such circuits that compute threshold functions approximately. Namely, the following theorem holds:
\begin{theorem}[\cite{ajtai}]\label{ajtai}
Let $\alpha\in(0,1)$. Then for any (constant) $\eps$ there exists a constant-depth and polynomial-size circuit $C$ such that $C(x)=0$ if the fraction of $1$'s in $x$ is less than $\alpha-\eps$ and $C(x)=1$ if the fraction of $1$'s in $x$ is greater than $\alpha+\eps$.
\end{theorem}
Note that nothing is promised if the fraction of $1$'s is between $\alpha-\eps$ and $\alpha+\eps$. So, the fact that $C(s)=0$ guarantees only that the fraction of $1$'s is at most $\alpha+\eps$, and $C(s)=1$~--- that it is at least $\alpha-\eps$.

\section{Muchnik's theorem}\label{muchnik}
\subsection{Subject overview}
An.~Muchnik's theorem~\cite{muchnik-codes} on conditional Kolmogorov complexity states that:

\begin{theorem}\label{main-theorem}
Let $a$ and $b$ be two binary strings such that $\KS(a)<n$ and $\KS(a|b)<k$.
Then there exists a string~$p$ such that

\textbullet\ $\KS(a|p,b) = O(\log n)$;

\textbullet\ $\KS(p) \le  k+O(\log n)$;

\textbullet\ $\KS(p|a) =  O(\log n)$.
\end{theorem}
The constants hidden in $O(\log n)$ do not
depend on $n,k,a,b,p$.

Informally, this theorem says that there exists a program $p$
that transforms $b$ to $a$, has the minimal possible complexity
$\KS(a|b)$ (up to a logarithmic term) and, moreover, can be easily
obtained from $a$. (The last requirement is crucial, otherwise
the statement trivially reformulates the definition
of conditional Kolmogorov complexity.)

Several proofs of this theorem are known. All of them rely on the existence of some combinatorial objects. The original proof in~\cite{muchnik-codes} leans upon the existence of bipartite graphs with specific expander-like property. Two proofs by Musatov, Romashchenko and Shen~\cite{tocs} use extractors and graphs allowing large on-line matchings. Explicit constructions of extractors provide variants of Muchnik's theorem for resource-bounded Kolmogorov complexity. Specifically, the following theorem is proven in~\cite{tocs}:
\begin{theorem}\label{muchnik-space}
Let $a$ and $b$ be binary strings of length $n$, and $k$ and $s$ be integers
such that $\KS^{s}(a|b)<k$. Then
there exists a binary string $p$, such that

\begin{itemize}
\item[\textbullet] $\KS^{O(s)+\poly(n)}(a|p,b) = O(\log^3 n);$

\item[\textbullet] $\KS^{s}(p)\le k+O(\log n);$

\item[\textbullet] $\KS^{\poly(n)}(p|a)= O(\log^3 n),$

\end{itemize}
where all constants in $O$- and $\poly$-notations depend only on
the choice of the optimal description method.
\end{theorem}
One cannot reduce residuals to $O(\log n)$ since all explicit extractors with such seed length output only $(1-\alpha)k$ bits, but $p$ is taken as an output of an extractor and should have length $k$. 
However, an application of our derandomization method will decrease conditional complexities from $O(\log^3 n)$ to $O(\log n)+O(\log\log s)$ at the price of increasing the space limit in the last complexity from $\poly(n)$ to $O(s)+\poly(n)$. If $s$ is polynomial in $n$ then all space limits become also polynomial and all complexity discrepancies become logarithmic.

\subsection{Proof overview}
Before we proceed with the detailed proof, let us present its high-level description. The main idea is the same as in all known proofs: $p$ is a fingerprint (or hash value) for $a$ constructed in some specific way. This fingerprint is chosen via some underlying bipartite graph. Its left part is treated as the set of all strings of length $n$ (i.e., all possible $a$'s) and its right part is treated as the set of all possible fingerprints. To satisfy the last condition each left-part vertex should have small outdegree, since in that case the fingerprint is described by its ordinal number among $a$'s neighbours. To satisfy the first condition each fingerprint should have small indegree from the strings that have low complexity conditional on $b$ (for arbitrary $b$). 

If resources are unbounded then the existence of a graph satisfying all conditions may be proven by the probabilistic method and the graph itself may be found by brute-force search. In the resource-bounded case we suggest to replace a random graph by a pseudo-random one and prove that it still does the job. The proof proceeds in several steps. Firstly, in Sect.~\ref{sect-lowcon} we define the essential graph property needed to our proof. We call this property low-congestion. It follows from the extractor property, but not vice versa. Secondly, in Sect.~\ref{sect-space-enum} we specify the instrumental notion of space-bounded enumerability and prove some lemmas about it in connection to low-congesting graphs. Next, in Sect.~\ref{muchnik-derandom} we employ these lemmas to prove that the low-congesting property is testable by small circuits. Hence, by applying the main principle (lemma~\ref{mainprinciple}) this property is satisfied for pseudo-random graphs produced by the NW generator as well as for truly random ones. Moreover, a seed producing a graph with this property may be found in polynomial space. Finally, in Sect.~\ref{sect-proof} we formulate our version of Muchnik's theorem and prove it using the graph obtained on the previous step. I.e., we describe the procedure of choosing a fingerprint in this graph and then calculate all complexities and space requirements and assure that they do not exceed their respective limitations.

\subsection{Low-congesting graphs}\label{sect-lowcon}
In fact, the proof in~\cite{tocs} does not use the extractor property, but employs only its corollary. In this section we accurately define this corollary in a way that allows derandomization.

Fix some bipartite graph with parameters ($n$, $m$, $d$) and an integer $k<n$. Let there be a system $\mathcal{S}$ of subsets of its left part with the following condition: each $S\in\mathcal{S}$ contains less than $2^k$ vertices, and the whole system $\mathcal{S}$ contains $2^{O(n)}$ sets (note that there are $\sum_{i=1}^{2^k}C_{2^n}^i>2^{O(n)}$ different sets of required size, so the last limitation is not trivial). We refer to such systems as to ``\emph{relevant}'' ones. Having fixed a system $\mathcal{S}$, let us call the sets in it \emph{relevant} also.

\begin{lemma}\label{relevant-sets}
Let $\mathcal{S}_k=\{S\subset\{0,1\}^n\mid\exists b\,\exists s\, |b|=n\ \text{and}\ S=\{x\mid \KS^s(x|b)<k\}\}$. Then the system $\mathcal{S}_k$ is relevant.
\end{lemma}
\begin{proof}
By a standard counting argument, each set $S\in\mathcal{S}_k$ contains less than $2^k$ elements. If $b$ is fixed then the described sets are expanding while $s$ is rising. Since the largest set is smaller than $2^k$, there are less than $2^k$ different sets for a fixed $b$. Since there are $2^n$ different strings $b$, the total size of $\mathcal{S}_k$ is bounded by $2^n2^k=2^{O(n)}$.
\end{proof}
Considering a modification of the upper system $\mathcal{S}_{k,\bar s}=\{S\subset\{0,1\}^n\mid\exists b\,\exists s<\bar s\, S=\{x\mid \KS^s(x|b)<k\}\}$, one may note that it is relevant as well.

Now fix some relevant system $\mathcal{S}$ and take an arbitrary set $S$ in it. Call the \emph{$\alpha$-clot} for $S$ the set of right-part vertices that have more than $\alpha DK/M$ neighbours in $S$ (that is, at least $\alpha$ times more than on average). Call a vertex $x\in S$ \emph{$\alpha$-congested} (for $S$) if all its neighbours lie in the $\alpha$-clot for $S$. Say that $G$ is \emph{($\alpha$,~$\beta$)-low-congesting} if there are less than $\beta K$ $\alpha$-congested vertices in any relevant $S$.

Following~\cite{fortnow}, we prove the next lemma:
\begin{lemma}\label{fortnow-new}
Let $G$ be an extractor graph with parameters $n$, $m$, $d$, $k$, $\eps$. Then for any $\alpha>1$ the graph $G$ is ($\alpha$,~$\frac{\alpha}{\alpha-1}\eps$)-low-congesting.
\end{lemma}
\begin{proof}
In fact in an extractor graph there are less than $\frac{\alpha}{\alpha-1}\eps K$ $\alpha$-congested vertices in any set $S$ of size $K$, not only in relevant ones. We may treat $S$ as an arbitrary set of size exactly $K$: since a congested vertex in a subset is also a congested vertex in the set, an upper bound for the number of congested vertices in the set holds also for a subset.

Let $Y$ be the $\alpha$-clot for $S$, and $|Y|=\delta M$. Then the fraction $|Y|/M$ of vertices in $Y$ equals $\delta$ and the fraction $E(S,Y)/E(S,M)$ of edges directing from $S$ to $Y$ is greater than $\alpha\delta$ (by the definition of clot). A standard counting argument implies only $\delta\le\frac1{\alpha}$, but by the extractor property $E(S,Y)/E(S,M)-|Y|/M<\eps$, so $(\alpha-1)\delta<\eps$, i.e. $\delta<\frac1{\alpha-1}\eps$. Next, let $T\subset S$ be the set of $\alpha$-congested vertices in $S$, and $|T|=\beta K$. All edges from $T$ direct to vertices in $Y$, so at least $D|T|=\beta DK$ edges direct from $S$ to $Y$. In other words, the fraction of edges from $S$ to $Y$ is at least $\beta$. By the extractor property it must differ from the fraction of vertices in $Y$ (that equals $\delta$) by at most $\eps$. So, $\beta<\delta+\eps<\frac1{\alpha-1}\eps+\eps=\frac{\alpha}{\alpha-1}\eps$. Putting it all together, the number of $\alpha$-congested vertices in any relevant $S$ is less than $\frac{\alpha}{\alpha-1}\eps K$, so the graph is ($\alpha$,~$\frac{\alpha}{\alpha-1}\eps$)-low-congesting, q.e.d.
\end{proof}

\subsection{Space-bounded enumerability}\label{sect-space-enum}
Say that a system $\mathcal{S}$ is \emph{enumerable} in space $q$ if there exists an algorithm with two inputs $i$ and $j$ working in space $q$ that either outputs the $j$th element of the $i$th set from $\mathcal{S}$ or indicates that at least one of the inputs falls out of range. Note that for a polynomial space bound enumeration is equivalent to recognition: if one may enumerate a set then one may recognize whether a given element belongs to it by sequentially comparing it to all enumerated strings, and if one may recognize membership to a set then one may enumerate it by trying all possible strings and including only those accepted by the recognition algorithm. Only small auxiliary space is needed to perform these modifications.

\begin{lemma}\label{relevant-enum}
The system $\mathcal{S}_{k,\bar s}=\{S\subset\{0,1\}^n\mid\exists b\,\exists s<\bar s\ |b|=n\ \text{and}\ S=\{x\mid \KS^s(x|b)<k\}\}$ is enumerable in space $O(\bar s)+\poly(n)$.
\end{lemma}
\begin{proof}
Assume firstly that a set $S$ is given (by specifying $b$ and $s<\bar s$) and show how to enumerate it. Look through all programs shorter than $k$ and launch them on $b$ limiting the space to $s$ and counting the number of steps. If this number exceeds $c^s$ (for some constant $c$ depending only on the computational model) then the current program has looped. If the looping is detected or if the program exceeds the space limit it is terminated and the next one is launched. Otherwise, if the program produces an output, then a check whether it has not been produced by any previous program is performed. This check proceeds as follows: store the result, repeat the same procedure for all previous programs and compare their results with the stored one. If no result coincides then include the stored result in the enumeration, otherwise skip it and in both cases launch the next program. Emulation of a program requires $O(s)$ space and no two emulations should run in parallel. All intermediate results are polynomial in $n$, so a total space limit $O(\bar s)+\poly(n)$ holds.

Specifying a set $S$ by $b$ and $s$ is not reasonable because a lot of values of $s$ may lead to the same set. (And if $\bar{s}=2^{\omega(n)}$ then the number of possible indexes $(b,\,s)$ exceeds the limit $2^{O(n)}$ on the size of $\mathcal{S}$). Instead, call a limit $s$ pivotal if $\{x\mid \KS^s(x|b)<k\}\ne\{x\mid \KS^{s-1}(x|b)<k\}$ and consider only pivotal limits in the definition of $\mathcal{S}_{k,\bar s}$. Clearly, the latter modification of definition does not affect the system itself. The advantage is that there are less than $2^k$ pivotal limits, and the $i$th pivotal limit $s_i$ may be found algorithmically in space $O(s_i)+\poly(n)$. Indeed, it is sufficient to construct an algorithm recognizing whether a given limit is pivotal, and the latter is done by a procedure similar to one described in the first paragraph: try all possible programs in space $s$, and in case they produce an output, check whether it is produced by any program in space $s-1$. If at least one result is new then $s$ is pivotal, otherwise it is not. Putting all together, a set $S$ in $\mathcal{S}_{k,\bar s}$ is defined by a word $b$ and the ordinal number $i$ of a pivotal space limit $s_i$. Knowing these parameters, one may enumerate it in space $O(\bar s)+\poly(n)$. Only $O(n)$ additional space is needed to look over all possible values of $b$ and $i$, so the declared bound is meeted.
\end{proof}

The next lemma indicates that if a system $\mathcal{S}$ is enumerable in small space, then the same holds for the system of congested subsets of its members. We call a bipartite graph \emph{computable} in space $q$ if there exists an algorithm working in space $q$ that receives an index of a left-part vertex and an index of an incident edge and outputs the right-part vertex this edge directs to.
\begin{lemma}\label{congested-enum}
Let $\mathcal{S}$ be a system of relevant sets enumerable in space $s$, $G$ be a bipartite graph computable in space $q$ and $\alpha>1$ be a (rational) number. Then the system $\Con_\alpha\mathcal{S}=\{T\mid\exists S\in\mathcal{S}\ \text{for\ which\ }T\text{\ is\ the\ set\ of $\alpha$-congested\ vertices}\}$ is computable in space $O(\max\{s,q\})+\poly(n)$. Moreover, the iteration $(\Con_\alpha)^r\mathcal{S}$ is also computable in space $O(\max\{s,q\})+\poly(n)$ with constants in $O$- and $\poly$- notations not depending on $r$, but possibly depending on $\alpha$.
\end{lemma}
\begin{proof}
Since for space complexity enumeration and recognition are equivalent, it is sufficient to recognize that a vertex $x$ is $\alpha$-congested. The recognizing algorithm works as follows: having a set $S$ and a vertex $x\in S$ fixed, search through all neighbours of $x$ (this requires space $O(q)+\poly(n)$) and for each neighbour check whether it lies in the  $\alpha$-clot for $S$. If all neighbours do then $x$ is congested, otherwise it is not. Having a neighbour $y$ fixed, the check is performed in the following way: enumerate all members of $S$ (using $O(s)+\poly(n)$ space), for each member search through all its neighbours (using $O(q)+\poly(n)$ space) and count the number of these neighbours coinciding with $y$. Finally, compare this number with the threshold $\alpha DK/M$. Note that no two computations requiring space $s$ or $q$ run in parallel and all intermediate results need only polynomial space, so the total space requirement is $O(\max\{s,q\})+\poly(n)$, as claimed. Note also that $O$- notation is used only due to the possibility of computational model change, not because of the necessity of looping control. If a computational model is fixed then the required space is just $\max\{s,q\}+\poly(n)$.

The last observation is crucial for the ``moreover'' part of lemma. Indeed, by a simple counting argument the fraction of $\alpha$-congested vertices in $S$ is at most $1/\alpha$. That is why after at most $\log_{\alpha}2^n=O(n)$ iterations the set of congested variables becomes empty. Each iteration adds $\poly(n)$ to the space requirement, so the overall demand is still $\max\{s,q\}+\poly(n)$ (with greater polynomial), as claimed.
\end{proof}

\subsection{Derandomization}\label{muchnik-derandom}
In this section we show that, firstly, the low-congesting property may be (approximately) recognized by $2^{O(n)}$-sized constant-depth circuits, secondly, that there are low-congesting graphs in the output of Nisan-Wigderson pseudo-random generator and, thirdly, that one can recognize in polynomial space whether the NW-generator produces a low-congesting graph on a given seed. Put together, the last two lemmas provide a polynomial-space algorithm outputting a seed on which the NW-generator produces a low-congested graph.

Let us encode a bipartite graph with parameters ($n$, $m$, $d$) by a list of edges. The length of this list is $2^n2^dm$: for each of $2^n$ left-part vertices we specify $2^d$ neighbours, each being $m$-bit long.
\begin{lemma}\label{muchnik-circuit}
Let $\mathcal{G}$ be the set of all bipartite graphs with parameters ($n$, $m$, $d$) encoded as described above. Let $k$ be an integer such that $1<k<n$, and let $\eps$ be a real positive number. Then there exists a circuit $C$ of size $2^{O(n)}$ and constant depth defined on $\mathcal{G}$, such that:
\begin{itemize}
\item If $G$ is a ($k$, $\eps$)-extractor then $C(G)$=1;
\item If $C(G)=1$ then $G$ is a $(2.01$, $2.01\eps)$-low-congesting graph for $\mathcal{S}_k$ from lemma~\ref{relevant-sets}.
\end{itemize}
\end{lemma}

\begin{proof}
We build a non-uniform circuit, so we may assume that $\mathcal{S}_k$ is given. We construct a single circuit approximately counting the number of congested vertices in a given set, then replicate it for each relevant set and take conjunction. Since there are less than $2^{n+k}$ relevant sets, this operation keeps the size of circuit being $2^{O(n)}$. We proceed by constructing a circuit for a given set $S$. A sketch of this circuit is presented on Fig.~\ref{circuit}.

The size of the input is $|S|\cdot2^d m$. We think of it as of being divided into $|S|$ blocks of $2^d$ segments of length $m$. We index all blocks by elements $x\in S$ and index all segments of the block $x$ by vertices $y$ incident to $x$. It is easy to see that there is a constant-depth circuit that compares two segments (that is, has $2m$ inputs and outputs 1 if and only if the first half of inputs coincides with the second half). On the first stage we apply this circuit to every pair of segments, obtaining a long 0-1-sequence. On the second stage we employ a counting circuit with $|S|D-1$ arguments that is guaranteed to output 1 if more than $2.01DK/M$ of its arguments are 1's and to output 0 if the number of 1's is less than $2DK/M$. By theorem~\ref{ajtai} there exists such a circuit of polynomial (in the number of arguments) size and constant depth. For all segments $y$ a copy of this circuit is applied to the results of the comparison of $y$ to all other segments. If $y$ lies in the $2.01$-clot then the respective copy outputs $1$, and if it outputs $1$, then $y$ lies in $2$-clot. On the third stage we take a conjuction of all second-stage results for the segments lying in the same block $x$. If this conjunction equals 1 then all images of $x$ lie in $2$-clot, that is, $x$ is $2$-congested. Conversely, if $x$ is $2.01$-congested then the conjunction equals 1. Finally, we utilize another counting circuit with $|S|$ inputs that outputs 0 if \textit{more} than $2.01\eps K$ of its inputs are 1's and outputs 1 if \textit{less} than $2\eps K$ of its inputs are 1's. This circuit is applied to all outputs of the third stage. If the final result is 1 then less than $2.01\eps K$ elements of $S$ are $2.01$-congested; and if less than $2\eps K$ elements of $S$ are $2$-congested then the final result is 1.
\begin{figure}[t]\label{circuit}
\hbox{\hskip -13mm \includegraphics{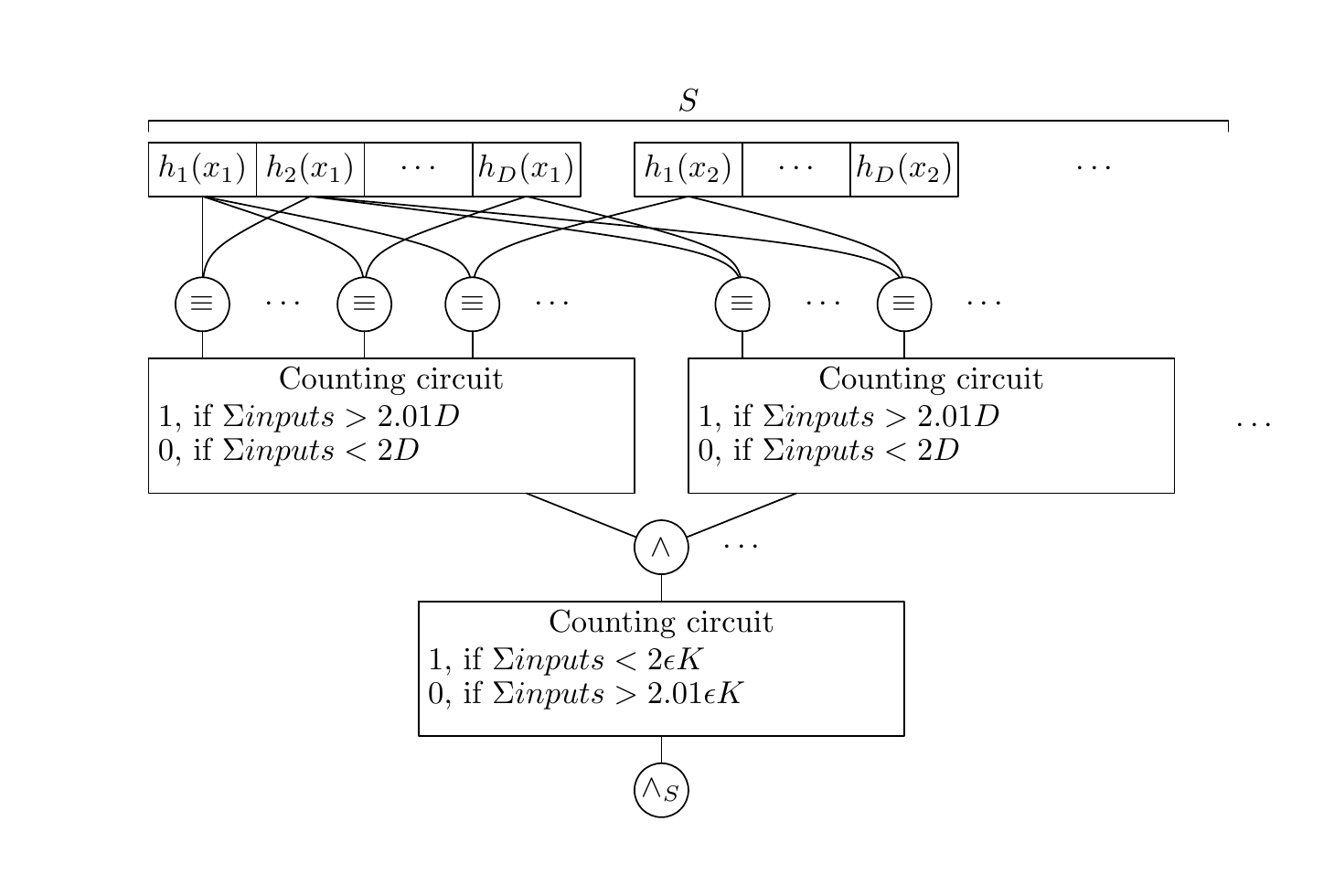}}
\caption{A sketch of the circuit for fixed $S$.}\end{figure}

The last claim holds for any relevant $S$. On the very last stage we take a conjunction of results for all $S$. If this conjunction is 1 then less than $2.01\eps K$ elements in each $S$ are $2.01$-congested, and if $G$ is a ($k$, $\eps$)-extractor then by lemma~\ref{fortnow-new} less than $2\eps K$ elements in each $S$ are $2$-congested and hence this conjunction equals 1. 

Since there are at most $2^{O(n)}$ gates on every stage and each stage has constant depth, the overall circuit has also $2^{O(n)}$ gates and constant depth, as claimed.
\end{proof}

\begin{lemma}\label{lowcon-derandom}
Let $n$, $m=k$, $d=O(\log n)$ and $\eps$ be such parameters that a random bipartite graph with parameters ($n$, $d$, $m$) is a ($k$, $\eps$)-extractor with constant probability $p>0$. Let $q$ be the depth of the circuit from the previous lemma. Let $NW\colon\{0.1\}^l\to\{0,1\}^N$, where $l=O(\log^{2q+6})$ and $N=2^n2^d m$, be the Nisan-Wigderson generator from corollary~\ref{nw-modified}. Then $\prob_u\{C(NW(u))=1\}>\frac p2$ for sufficiently large $n$, where $C$ is the circuit from the previous lemma.
\end{lemma}
\begin{proof}
This is a straightforward application of lemma~\ref{mainprinciple}. By the previous lemma if a graph $G$ is an extractor then $C(G)=1$, so $\prob_G\{C(G)=1\}\ge p$. Since $C$ is a constant-depth circuit, the property $C(G)=1$ is tautologically testable by a constant-depth circuit. By lemma~\ref{mainprinciple} $\prob_u\{C(NW(u))=1\}>\frac p2$, q.e.d.
\end{proof}

Consider the following problem $\mathcal{R}$: find a string $u\in\{0,1\}^l$ such that the graph $NW(u)$ is ($2.01$, $2.01\eps K$)-low-congesting with respect to $\mathcal{S}_{k,\,s}$.

\begin{lemma}\label{finding-good-seed}
The problem $\mathcal{R}$ is solvable in space $O(s)+\poly(n)$.
\end{lemma}
\begin{proof}
The existence of a solution follows from the previous lemma. Since we care only about the space limit, the search problem may be replaced by the corresponging recognition problem. The space bound for the latter one arises from corollary~\ref{nw-modified}, lemma~\ref{relevant-enum} and lemma~\ref{congested-enum}. Indeed, by corollary~\ref{nw-modified} the graph $NW(u)$ is computable in polynomial space, and by lemma~\ref{relevant-enum} the system $\mathcal{S}_{k,\,s}$ is enumerable in space $O(s)+\poly(n)$. Hence by lemma~\ref{congested-enum} the system $\Con_{2.01}\mathcal{S}_{k,\,s}$ is also enumerable in space $O(s)+\poly(n)$, therefore one may easily check whether each set in $\Con_{2.01}\mathcal{S}_{k,\,s}$ contains less than $2.01\eps K$ elements, thus solving the recognition analogue of $\mathcal{R}$. Only polynomial extra space is added on the last step.
\end{proof}

\subsection{Proof of the theorem}\label{sect-proof}
Now we proceed by formulating and proving our version of Muchnik's theorem.

\begin{theorem}
Let $a$ and $b$ be binary strings of length less than $n$, and $s$ and $k$ be numbers such that $\KS^s(a|b)<k$. Then there exists a binary string $p$, such that
\begin{itemize}
\item[\textbullet] $\KS^{O(s)+\poly(n)}(a|p,b)=O(\log\log s+\log n)$;
\item[\textbullet] $\KS^s(p)\le k+O(\log n)$;
\item[\textbullet] $\KS^{O(s)+\poly(n)}(p|a)=O(\log\log s+\log n)$,
\end{itemize}
where all constants in $O$- and $\poly$-notations depend only on the choice of the optimal description method.
\end{theorem}
\begin{proof}
Basically the proof proceeds as the respective proof of theorem~\ref{muchnik-space} in~\cite{tocs}. Here we replace an explicitly constructed extractor by a pseudorandom graph described in Sect.~\ref{muchnik-derandom}.

Let $\eps$ be a small constant and let $d=O(\log n)$ be such that a random bipartite graph with parameters ($n$, $m=k$, $d$) is a ($k$,~$\eps$)-extractor with probability greater than some positive constant $\mu$. Let $l$ be a parameter and $NW$ be a function from lemma~\ref{lowcon-derandom}. By lemmas~\ref{lowcon-derandom} and~\ref{muchnik-circuit} the output of $NW$ is a ($2.01$, $2.01\eps K$)-low-congesting graph for $\mathcal{S}_k$ with probability at least $\mu/2$. Applying the program from lemma~\ref{finding-good-seed}, one may find in $O(s)+\poly(n)$ space a seed $u$ for which $NW(u)$ is a ($2.01$, $2.01\eps K$)-low-congesting graph for $\mathcal{S}_{k,\,s}$. This $u$ has low complexity: to perform this search one needs to know parameters $n$, $k$ and $l=\poly(n)$, that is, $O(\log n)$ bits, and the space bound $s$, that requires $\log s$ bits. The last number may be reduced to $\log\log s$, because the space bound $s$ may be replaced by the least power of 2 exceeding $s$ keeping the needed space to be $O(s)+\poly(n)$. For what follows, fix this seed $u$ and the graph $G=NW(u)$.

By definition of the low-congesting property the number of $2.01$-congested vertices in the set $S=\{x\mid \KS^s(x|b)<k\}$ does not exceed $2.01\eps K$. Clearly, $a$ belongs to $S$. Firstly suppose that it is not $2.01$-congested. Then it has a neighbour outside the $2.01$-clot for $S$. This neighbour may be taken as $p$. Indeed, the length of $p$ equals $m=k$, hence its complexity is also less than $k+O(1)$. To specify $p$ knowing $a$, one needs to construct the graph $NW(u)$, for which only $O(\log n)+O(\log\log s)$ bits are necessary, and to know the number of $p$ among $a$'s neighbours, which is at most $d=O(\log n)$. Finding $a$ given $b$ and $p$ proceeds as follows: construct $G$, enumerate $S$ and choose the specified preimage of $p$ in $S$. Then $a$ is determined by the information needed to construct $G$ ($O(\log n+\log\log s)$ bits), information needed to enumerate $S$ (that is, $k$, $\log s$ and $b$) and the number of $a$ among preimages of $p$ (since $p$ is not in the $2.01$-clot, there are not more than $2.01DK/M=2.01D$ preimages, so $O(d)=O(\log n)$ bits are required). Summarizing, we get $O(\log n)+O(\log\log s)$ bits and $O(s)+\poly(n)$ space needed. Note that the fact that $m=k$ is crucial here: in the case $m=(1-\alpha)k$ the number of required bits would increase by $\alpha k$ and exceed the bound.

Now we turn to the case where $a$ \textit{is} $2.01$-congested. By lemma~\ref{congested-enum} the set $\Con_{2.01}S$ of $2.01$-congested vertices is enumerable in $O(s)+\poly(n)$ space. Take parameters $n_1=n$, $m_1=k_1=\log K_1=\log(2.01\eps K)$, $d_1=O(\log n)$ and $\eps_1=\eps$ such that an extractor with these parameters exists with probability greater than $\mu$. Lemmas~\ref{muchnik-circuit}, \ref{lowcon-derandom} and \ref{finding-good-seed} are applicable to the new situation, so we may find a new $u_1$ such that the graph $G_1=NW(u_1)$ is ($2.01$,~$2.01\eps K_1$)-low-congesting for $\Con_{2.01}\mathcal{S}_{k,\,s}$ with probability at least $\mu/2$. By assumption, $a$ belongs to the set $\Con_{2.01}S$. If it is not $2.01$-congested in the new set then we choose $p$ analogously to the initial situation. Otherwise we reduce the parameters again and take a low-congesting graph for $\Con_{2.01}^2\mathcal{S}_{k,\,s}$, and so on. By the ``moreover'' part of lemma~\ref{congested-enum}, this reduction may be performed iteratively for arbitrary number of times keeping the space limit to be $O(s)+\poly(n)$.

The total number of iterations is less than $\log_{1/2.01\eps}k=O(\log n)$. Finally $p$ is defined as a neighbour of $a$ not lying in the $2.01$-clot in graph $G_i$. To find $p$ knowing $a$ one needs to know $i$ and the same information as on the upper level. To find $a$ knowing $p$ and $b$ also only specifying $i$ is needed besides what has been specified on the upper level. So, all complexities and space limits are as claimed and the theorem is proven.
\end{proof}

\begin{center}
        \textbf{Acknowledgments}
\end{center}
        \nopagebreak

I want to thank my colleagues and advisors Andrei Romashchenko, Alexander Shen and Nikolay Vereshchagin for stating the problem and many useful comments. I also want to thank three anonymous referees for careful reading and precise comments. I am grateful to participants of seminars in Moscow State University and Moscow Institute for Physics and Technology for their attention and thoughtfulness.

\end{document}